\documentclass[conference]{IEEEtran}
\usepackage{cite}
\ifCLASSINFOpdf
  \usepackage[pdftex]{graphicx}
\else
  \usepackage[dvips]{graphicx}
  \graphicspath{{../eps/}}
  \DeclareGraphicsExtensions{.eps}
\fi
\usepackage{tikz}
\usepackage{tikzscale}
\usepackage{pgf}
\usepackage{tikz-network}

\usepackage[T1]{fontenc} 
\usepackage{physics}
\usepackage{hyperref}
\usepackage{booktabs,tabularx}
\usepackage{authblk}

\hypersetup{
    bookmarks=true,         % show bookmarks bar?
    unicode=false,          % non-Latin characters in Acrobat’s bookmarks
    pdftoolbar=true,        % show Acrobat’s toolbar?
    pdfmenubar=true,        % show Acrobat’s menu?
    pdffitwindow=false,     % window fit to page when opened
    pdftitle={},    % title
    pdfauthor={},     % author
    pdfsubject={},   % subject of the document
    pdfcreator={},   % creator of the document
    pdfproducer={},  % producer of the document
    pdfkeywords={} {} {}, % list of keywords
    pdfnewwindow=true,      % links in new window
    colorlinks=false,       % false: boxed links; true: colored links
    linkcolor=red,          % color of internal links
    citecolor=green,        % color of links to bibliography
    filecolor=magenta,      % color of file links
    urlcolor=cyan           % color of external links
}

\usepackage{amsmath,amssymb, amsthm, bm}

\newtheorem{proposition}{Proposition}
\newtheorem{observation}{Observation}
\newtheorem{procedure}{Procedure}
\usepackage{algorithmic}
\usepackage{enumitem}

\usepackage{array}
\ifCLASSOPTIONcompsoc
  \usepackage[caption=false,font=normalsize,labelfont=sf,textfont=sf]{subfig}
\else
  \usepackage[caption=false,font=footnotesize]{subfig}
\fi

\usepackage{stfloats}
\usepackage{url}
\begin{document}
% \title{How Entanglement-Enabled Connectivity Can Boost the Next-Generation Communication and Computing — A Tutorial}

\title{Lossless Address Coding for Quantum Networks}

%\author{\IEEEauthorblockN{Dick Maryopi}
%\thanks{The authors are with the Department of Electronic Engineering, University of York, Heslington, York, UK. email: dm1110@york.ac.uk.}
%}

\vspace{-0.25cm}
\author[]{Dick~Maryopi}

\affil[]{\normalsize University of Naples Federico II,  Naples, 80125 Italy.\\
dick.maryopi@unina.it}
\vspace{-0.25cm}

%\thanks{}
%\renewcommand{\Authfont}{\normalsize}
\maketitle
\begin{abstract}
As quantum systems advance toward interconnected architectures, the ability to identify nodes, manage resources, and support network‑level functions becomes increasingly critical. In this work, we propose a lossless source coding scheme for addressing in quantum networks that enables compact, hierarchical, and coherently processable quantum address states. Specifically, we introduce a prefix–suffix address space and develop an isometric hierarchical encoder-decoder that guarantees unique decodability. We further provide a practical Huffman‑based procedure that embeds prefix‑free, length‑eigenstate codewords into the address space, thereby preserving isometry. The scheme is particularly designed for networks with hierarchical structure, heterogeneous cluster sizes, and configurable address assignment to accommodate dynamic network conditions. A numerical example on a 13‑node network demonstrates that the proposed hierarchical encoding scheme is feasible and achieves perfect fidelity. This work establishes a rigorous connection between source coding theory and quantum network design, offering a practical framework towards scalable and coherent quantum addressing.
\end{abstract}

% Note that keywords are not normally used for peer-reviewed papers.

\IEEEpeerreviewmaketitle
%===================================================================================================================
\section{Introduction}%% It should be informative, not so long and not too short %===================================================================================================================
We are currently witnessing a crucial transition in the field of quantum technologies, in which the development of quantum computing is progressing from isolated systems toward interconnected systems that form quantum networks. This networking of quantum systems not only enhances the overall computational capabilities but also enables a wide range of applications, including secure quantum communication, high-precision sensing, and many other distributed quantum information processing tasks. Given this potential, it is envisioned that quantum networks will evolve into a larger-scale, and globally interconnected architecture, referred to as the quantum internet \cite{Kimble2008, Dur2017, Wehner2018, Cacciapuoti2020,RFC9340}.

However, extending quantum systems to large-scale networks, particularly at a global level, is far from trivial and has been a challenging research effort. This is mainly due to the fact that the networks must satisfy fundamentally different physical constraints, such as coherence preservation, restriction of measurement, and the non-clonability of unknown quantum information. Consequently, most information-processing techniques in classical networks cannot be directly applied in the quantum setting. In addition, entanglement, which will serve as the fundamental resource in quantum networking, still faces inherent technical challenges in its generation, storage, and distribution, despite its beneficial properties that enable flexible networking. 

All of those factors lead to time-varying entanglement network topologies that are dependent on the availability of entanglement and its manipulation, thereby creating a highly dynamic condition. In addition, practical quantum networks should accommodate an increasing number of nodes and a wide range of network sizes. In such networks, identifying participating nodes, determining the source and destination of the communicating parties, allocating resources among nodes, and supporting many other functionalities become even more challenging. These tasks all rely on an addressing mechanism, where directing to an incorrect address may have a significant impact. Hence, addressing should operate in a lossless manner and remain valid as the network grows, despite the steady changes in the entanglement network's topology.

Until recently, only a limited number of studies have investigated the problem of addressing within quantum networks \cite{10299662, PhysRevResearch.4.043064, 11322738, MiguelRamiro2021SuperposedTasks,Pirker2026QuantumAddressing}. Among these works, the initial concepts of addressing were introduced in \cite{10299662} in the context of the quantum internet. It proposed a quantum-native addressing functionality that shifts from a classical location-aware addressing scheme to one built on an entangled overlay network via link augmentation. In contrast, \cite{PhysRevResearch.4.043064} proposed a classical–quantum hybrid frame structure that carries classical addresses for quantum data payloads. The work in \cite{11322738} extends quantum-native addressing to enable scalable, compact routing, in which the matching-and-forward operation is performed via quantum address splitting using Schrödinger’s oracle. Unlike \cite{11322738}, which uses a superposition of node identifiers to achieve scalability, the work in \cite{MiguelRamiro2021SuperposedTasks} uses addressing to perform the superposition of tasks. Meanwhile, \cite{Pirker2026QuantumAddressing} explores quantum addresses that can be in entangled states to facilitate quantum routing.  

While these works establish important mechanisms for quantum addressing, they do not fundamentally resolve the question of how the address should be represented so that it is flexible enough to accommodate the structure and dynamics of the networks while satisfying the quantum-physical constraints. In this work, we approach this question from an information-theoretic coding perspective. In this context, naming or labeling a set of objects is essentially a natural instance of a source coding task \cite{PhysRevA.65.032313}. Therefore, we use this coding-theoretic lens and show that the challenge of addressing in quantum networks can be viewed as a lossless coding problem, where address should be represented as decodable quantum states and processed coherently to remain meaningful under a dynamic environment.

To that end, we propose a lossless source coding scheme for quantum networks, particularly those with a hierarchical structures, to support scalability. The main contributions of this work can be summarized as follows:
\begin{enumerate}
    \item We provide a framework for addressing in quantum networks by treating node identification as an information-theoretic lossless source coding problem.
    \item We design an isometric hierarchical encoder and its adjoint decoder, and prove unique decodability and coherent (measurement-free) reconstruction.
    \item We develop a concrete encoding procedure that uses spectral diagonalization and Huffman-based prefix/suffix codes to produce codewords suitable for dynamical addressing.
    \item We demonstrate the scheme numerically on a toy example of a 13-node network, showing negligible recovery and isometry errors, and validate a hierarchical lossless encoding in practice.
\end{enumerate}

The remainder of the paper is organized as follows. Section \ref{system_model} introduces the system model considered in this paper. Section \ref{problem_formulate} formulates the addressing problem and states the design requirements. Section \ref{coding scheme section} develops the core contribution, that is, a hierarchical lossless‑coding construction, including the prefix–suffix structure and the isometric encoder/decoder. Section \ref{Section encoding procedure} presents a concrete, implementable encoding procedure (Huffman‑based spectral coding) and discusses practical design choices. Section \ref{numerical_example} illustrates the approach with a numerical example that validates the construction on a small network. Finally, Section \ref{conclusion} summarizes the results and outlines directions for future work.

\begin{figure}
    \centering
    \includegraphics[width=\columnwidth, height=0.3\textheight, keepaspectratio]{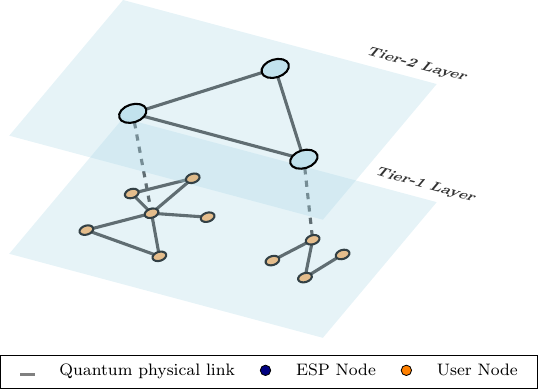}
    \caption{A simple hierarchical quantum network model considered in this paper. Network nodes are partitioned into two functional layers: 1) the tier‑2 layer serves as an entanglement service provider (ESP), and 2) the tier‑1 layer consists of user nodes organized into clusters that consume the distributed entanglement. }
    \label{fig:model}
\end{figure}
%===================================================================================================================
\section{System Model}\label{system_model}
%===================================================================================================================
We consider a quantum network represented as a graph
\begin{equation}
    G = (V, E),
\end{equation}
where $V$ is the set of nodes with cardinality $|V|=n$ and $E \subseteq V\times V$ is the set of quantum links that physically connect the nodes. The node set $V$ is bipartitioned into disjoint two-tier network layers as depicted in Fig.~\ref{fig:model}. Each layer is designed to serve different functions and has distinct characteristics \cite{11322738}.
In the tier-1 layer we have a set of user nodes that need to consume entanglement resources. It induces a subgraph given by
\begin{equation}
    G_\mathrm{u}=G[V_u]=(V_\mathrm{u}, E_\mathrm{u}), \text{ with } |V_\mathrm{u}|=n_\mathrm{u}.
\end{equation}
In the tier-2 layer we have a set of Entanglement Service Provider (ESP) nodes, which is responsible for distributing entanglement to user nodes in the tier-1 layer. The tier-2 layer induces a subgraph given by
\begin{equation}
    G_\mathrm{e}=G[V_\mathrm{e}]=(V_\mathrm{e}, E_\mathrm{e}), \text{ with } |V_\mathrm{e}|=n_\mathrm{e}.
\end{equation}
We require that all nodes be partitioned into these two layers in such a way that the partition satisfies
\begin{equation}
    V = V_\mathrm{e} \cup V_\mathrm{u}, \, V_\mathrm{e} \cap V_\mathrm{u} = \emptyset. %\text{ and }\, n=n_\mathrm{e}+n_{\mathrm{u}}.
\end{equation}

To ensure the network functions are scalable as the network grows, we further partition all nodes in $V$ into $n_\mathrm{e}$ clusters. We denote the resulting $j$-th cluster by the subgraph
\begin{align}
    &G_{j} = (V_{j}, E_{j}), \text{ for }\, j=1, \dots, n_\mathrm{e}, \,\\
    &\text{ such that } \nonumber \\ 
    &V=\bigcup_{j=1}^{n_\mathrm{e}} V_{j}, \text{ with } V_{j}\cap V_{j'}=\emptyset \text{ for all }j\neq j'. \nonumber
\end{align} 
Thus, every node $v_i\in V$ belongs to exactly one cluster $V_{j}$. To capture a more realistic modeling of a heterogeneous network scenario, we allow the clusters to have varying sizes $|V_{j}|$. 
In this work, we consider that each cluster $V_{j}$ is served by exactly one ESP node 
$v_{\mathrm{e}, j}\in V_\mathrm{e}=\{v_{\mathrm{e}, 1}, \dots, v_{\mathrm{e}, n_\mathrm{e}}\}$. 
Accordingly, each cluster $j$ is composed of the ESP node $v_{\mathrm{e},j}$ and a (possibly empty) set of user nodes $V_{\mathrm{u},j}\subseteq V_\mathrm{u}$, i.e.,
\[
    V_j = \{v_{\mathrm{e},j}\} \cup V_{\mathrm{u},j}.
\]
where $V_{\mathrm{u},j}$ denotes the set of user node in cluster $j$.
We require that each ESP node belongs to the cluster itself, i.e. $v_{e,j}\in V_j$ for all \(j\). This allows cluster label and its representative coincide, thereby simplifying the constructions of the hierarchical encoder and decoder in Sections \ref{coding scheme section}--\ref{Section encoding procedure}.

The structure induces the node-cluster association for any node $v\in V$, given by
\begin{equation}
    \mu: V \rightarrow \{1, \dots, n_\mathrm{e}\}, \,
    \mu(v)=j \text{ iff } v \in V_{j},    \label{node-cluster}
\end{equation}
which has a natural two-level hierarchy for addressing. Further, each user node $v_i\in V$ is assigned a unique identifier, referred to as an address, drawn from an address space $\mathcal{A}$. The specific properties and structure of $\mathcal{A}$ will be presented in Section \ref{coding scheme section}.

% In Section~\ref{coding scheme section}, we construct a nested code on the address space $\mathcal{A}$ with the following structure: 
% (i) all ESP nodes $v_{\mathrm{e},j}\in V_\mathrm{e}$ are assigned pairwise prefix-free codewords, and 
% (ii) all nodes in the same cluster $V_j$ share the same prefix (the codeword of $v_{\mathrm{e},j}$) and are distinguished by prefix-free suffixes within that cluster.

%===================================================================================================================
\section{Problem Formulation} \label{problem_formulate}
%====================================================

In the following, we aim to develop an addressing scheme that simultaneously satisfies the required quantum properties and the constraints imposed at the network level. From a quantum perspective, the addressing mechanism must support a representation that can be embedded into quantum states and coherently manipulated by quantum operations. This requirement arises because address information may accompany quantum data or entanglement in various quantum‑native network functionalities and, therefore, must remain compatible with the laws of quantum mechanics. 

At the same time, the addressing scheme must remain efficient at the network level. In particular, the address representation should be compact to avoid excessive overhead, while also supporting hierarchical identification and accommodating heterogeneous cluster sizes. In some cases, only the cluster identity may be needed, rather than the individual identity. Moreover, because quantum network topologies may evolve as entanglement links are created, consumed, or reconfigured, address assignment must be configurable and remain valid under dynamic conditions.
Hence, scalability and flexibility are essential properties of the addressing scheme.

Specifically, we seek to construct an injective addressing function
\begin{equation}
    f: V \rightarrow \mathcal{A},
\end{equation}
where $\mathcal{A}$ denotes the quantum address space. The injectivity of $f$
ensures that each node in the network can be uniquely identified from its associated quantum address.
To satisfy the network constraint, it is required that for each node $v_i \in V$, the corresponding quantum address $ \ket{c_i} \in \mathcal{A}$ can be expressed as
\begin{align}
    \ket{c_i}=\ket{\kappa_j \sigma_{ij}}, \label{hier_adrress_structure}
\end{align}
where $\ket{\kappa_j}$ denotes the prefix address for identifying the $j$-th cluster and $\ket{\sigma_{ij}}$ denotes the suffix address for identifying the $i$-th user node in the $j$-th cluster. 
Taking into account the heterogeneity in cluster sizes, the lengths of the prefix $\ell(\kappa_j)$and the length of the suffix $\ell(\sigma_{ij})$ are both treated as variable parameters. 

%===================================================================================================================
\section{Hierarchical Losslees Coding}\label{coding scheme section}
%===================================================================================================================
    To tackle the problem of addressing, we propose to design a scheme based on a lossless source coding approach. We model the node set $V$ as a discrete information source, which is described by an ensemble $V =\{p, \mathcal{V}\}$. Thereby, each node $v_i \in V$, indexed by \(i = 1, \dots, n\), is represented by a quantum message $\ket{v_i}\in \mathcal{V}$ that is prepared by the source with probability \(p(v_i)>0,\) for all $v_i \in V$. In practical scenarios, the probability \(p(v_i)\) can be specified, for example, based on empirical knowledge or on measurement data from the traffic observed at node \(v_i\). Since each node corresponds to a distinct physical entity, the specific source message $\ket{v_i}$ prepared for node $v_i$ is known a priori. A case that is referred to as visible coding \cite{HowardBarnum_2001}.

This source message $\ket{v_i}$ is then encoded into a codeword $\ket{c_i}$ by a lossless code $c$.
In this manner, each node \(v_i \in V\), can be assigned a distinct codeword $\ket{c_i}=\ket{c(v_i)} \in \mathcal{A}$ that represents its address, and the resulting set of codewords is uniquely decodable \footnote{The source message $\ket{v_i}$ attached to node $v_i$ can be seen as a MAC address, whereas the codeword $\ket{c_i(v_i)}$ plays the role of an IP address in a classical network.}. In a quantum setting, this assignment must be implemented as an isometric linear map satisfying
\begin{equation}
    \langle v|v'\rangle= \langle c(v)| c(v')\rangle, \text{ for all } v, v' \in V,
\end{equation}
which ensures that the encoding is physically realizable \cite{Nielsen_Chuang_2010, Wilde_2017}. To this end, we use the quantum lossless coding framework introduced by Boström \cite{PhysRevA.65.032313}, adapting and extending it to suit the specific requirements of our setting.

\subsection{Quantum Message Space}
To represent the source message $\ket{v_i}\in \mathcal{V}$ and construct the corresponding codewords $\ket{c_i}\in \mathcal{A}$, we work with a two-dimensional symbol space, that is, a Hilbert space $\mathcal{H}$ with the computational basis
\begin{equation}
    \mathcal{B} = \{\ket{x}: x \in\{0,1\}\}. \label{comp-basis}
\end{equation}
Any quantum message $\ket{v}$ of length $\ell$ has then a form of quantum bit string, which is composed of quantum bits via tensor multiplication
\begin{equation}
    \ket{v}=\ket{x_1 x_2 \dots x_\ell}:=\bigotimes_{k=1}^\ell\ket{x_k}, 
\end{equation}
where $\ket{x_k} \in \mathcal{B}$.
Such a quantum message lives in $\ell$-fold tensor product of Hilbert space
\begin{equation}
    \mathcal{H}^{\otimes \ell} := \bigotimes_{k=1}^\ell \mathcal{H} = \operatorname{span}(\mathcal{B}^\ell),
\end{equation}
equipped with the message basis
\begin{equation}
    \mathcal{B}^\ell = \left\{ \ket{\omega} := \bigotimes_{k=1}^\ell \ket{\omega_k} : \ket{\omega_k} \in \mathcal{B} \right\}.
\end{equation}
To support variable-length codewords, we need a larger space, for which \cite{PhysRevA.65.032313} recommends using the Fock space.

For practical deployment, we can restrict the Fock space to a bounded message space
\begin{equation}
    \mathcal{H}^{\oplus \ell_\mathrm{max}}:=\bigoplus_{\ell=0}^{\ell_\mathrm{max}} \mathcal{H}^{\otimes \ell},
\end{equation}
that is, the direct sum of all block message spaces up to length $\ell_\mathrm{max}$. In our setting, this space serves as the bounded quantum bitstring space. For our purpose, we therefore choose our address space $\mathcal{A}$ as a subspace of $\mathcal{H}^{\oplus \ell_\mathrm{max}}$.

\subsection{Encoder}
Let $\mathcal{V}\subset \mathcal{H}^{\otimes r}$ be the source message space of fixed length $r$, spanned by the message basis $\mathcal{B}_V^r$. 
In accordance with the scheme proposed by Boström \cite{PhysRevA.65.032313}, we define a lossless quantum code as an isometric embedding
\begin{equation}
    c: \mathcal{V} \longrightarrow \mathcal{A} \subset \mathcal{H}^{\oplus \ell_\mathrm{max}}, \label{general_code}
\end{equation}
which associates each source message $\ket{v}$ with a distinct address codeword $\ket{c(v)}$ such that
\begin{equation}
    \ket{c(v)}\neq\ket{c(v')} \quad \text{whenever} \quad \ket{v}\neq \ket{v'}. \label{injectivity}
\end{equation}
This encoding is realized by the isometric encoding operator
\begin{equation}
    C:= \sum_{\omega\in \mathcal{B}_V^r}\ket{c(\omega)}\bra{\omega}. \label{general_encoder}
\end{equation}
Recall that the operator $C$ from $\mathcal{V}$ to $\mathcal{A}$ is isometry if $C^{\dagger}C=I_{\mathcal{V}}$ and $CC^{\dagger}=\Pi_\mathcal{A}$, where $\Pi_\mathcal{A}$ is the projector to its range $\mathcal{A}$ \cite{Wilde_2017}.
Therefore, $C$ maps the orthonormal basis $\mathcal{B}_V^r$ of source messages onto a corresponding set of mutually orthogonal codewords in the address subspace $\mathcal{A}$.
% This orthonormal basis $\mathcal{B}_V^r$ can be constructed from the most probable source messages in the network. 
Although the encoder is built on basis $\mathcal{B}_V^r$, linearity ensures the codeword address of any node $v$ is obtained by applying the encoder to the source message for that node, given by
\begin{equation}
    \ket{c(v)} = C\ket{v}, \quad \text{for all } \ket{v} \in \mathcal{V}.
\end{equation}
Consequently, once the encoder has been specified, it can be used to perform lossless encoding of any source message in $\mathcal{V}$ of an arbitrary node $v \in V$.\footnote{This mechanism enables dynamic address assignment for a collection of nodes in a manner analogous to the Dynamic Host Configuration Protocol (DHCP) in classical communication networks.}

However, to accommodate the hierarchical addressing structure defined in (\ref{hier_adrress_structure}), the encoder (\ref{general_encoder}) must be decomposed into encoders that together form a hierarchical structure. To that end, we require a hierarchical code $\hat{c}$ so that the address codeword can be expressed as
\begin{equation}
    \ket{\hat{c}(v)} = \ket{\kappa(v) \sigma(v)},
\end{equation}
where $\kappa$ is a prefix code and $\sigma$ is a suffix code. Additionally, it is essential to ensure that the hierarchical encoder maintains the isometric property.
By recognizing that the source message space can be decomposed as 
\begin{equation}
    \mathcal{V} = \bigoplus_{j=0}^{n_e} \mathcal{V}_j, \quad \mathcal{V}_j := \mathrm{span}\{\ket{\omega} : \omega \in V_{j}\}, \label{SourceSpace_decomposition}
\end{equation}
we may characterize the prefix and suffix codes as follows.

\subsubsection{Prefix Code }
The prefix code $\kappa$ assigns the source message of any node $v\in V$ to a prefix address given by
\begin{equation}
    \kappa: \mathcal{V}_\mathrm{e} \longrightarrow \mathcal{A}_\mathrm{e} \subset \mathcal{H}^{\oplus \ell_\mathrm{max}},
\end{equation}
where $\mathcal{V}_\mathrm{e}\subset\mathcal{V}$ denotes the source message space associated with the ESP set $V_\mathrm{e}$, and $\mathcal{A}_\mathrm{e}$ is the code space of the prefix addresses. 
Our objective is to employ this code in such a way that all nodes belonging to the same cluster $V_j$ are assigned an identical prefix $\ket{\kappa_j}$, while nodes in distinct clusters are assigned mutually distinct prefixes. 
This prefix code $\kappa$ induces an isometric prefix-encoding operator
\begin{equation}
    {K}:= \sum_{\omega \in \mathcal{B}_{V_\mathrm{e}}} \ket{\kappa(\omega)}\bra{\omega}, \label{Prefix_Encoder_Operator}
\end{equation}
which maps basis message $\ket{\omega}$ to its cluster prefix. Thus, we can express
\begin{equation}
\ket{\kappa_{\mu(v)}}=\ket{\kappa(v)}={K}\ket{v},
\end{equation}
where $\mu$ denotes the node–cluster mapping defined in (\ref{node-cluster}), which returns the index $j$ whenever the node $v$ belongs to the cluster $V_j$. By assuming $\ket{v}\in\mathcal{V_\mathrm{e}}$ is normalized, we have $\Vert K\ket{v} \Vert = \Vert \ket{\kappa(v)}\Vert=1$.

\subsubsection{Suffix Code}
We now turn to the characterization of the suffix code $\sigma$, which provides an intra-cluster identifier that distinguishes nodes within the same cluster.  
For each cluster $V_j$, we introduce a cluster-specific suffix code
\begin{equation}
    \sigma_j : \mathcal{V}_j \longrightarrow \mathcal{A}_j \subset \mathcal{H}^{\oplus \ell_{j}},
\end{equation}
where $\mathcal{A}_j$ is the code space of the the suffix address in cluster $j$ and $\ell_{j}$ denotes the maximum length of suffix address in cluster $j$. The suffix mapping is defined as
\begin{equation}
\sigma_j(v) :=
\begin{cases}
\sigma_{ij}, & v \in V_{\mathrm{u},j},\\
\varnothing, & v = v_{\mathrm{e},j},
\end{cases} \label{suffix assgnment}
\end{equation}
such that a suffix address $\sigma_{ij}$ is assigned whenever the node $v$ belongs to cluster $j$ and is not the ESP node, whereas no address is assigned when $v$ is the ESP node $v_{\mathrm{e},j}$. Here, equation (\ref{suffix assgnment}) defines the abstract intra-cluster suffix assignment, which is formulated without imposing any operational constraints.
% In Section V we show that, for a concrete realization in the variable‑length Fock space, either prefixes must be implemented as length eigenstates (prefix‑free) or an explicit framing/delimiter or side channel must be provided; otherwise the empty suffix should be used only for singleton clusters. See Section V for operational details.
% Let $\ell(\sigma_{i,j})$ denote the length of the suffix assigned to node $v_{i} \in V_j$.  
% Then, for each cluster $j$, the intra-cluster Kraft inequality must hold:
% \begin{equation}
%     \sum_{v_i \in V_j} 2^{-\ell(\sigma_{i,j})} \le 1.
%     \label{eq:KraftSuffix}
% \end{equation}
% This condition guarantees that the suffixes within each cluster are uniquely decodable and that the hierarchical concatenation
% \[
% \ket{\kappa(v)}\otimes\ket{\sigma(v)}
% \]
% remains prefix-free at the global level.

This suffix code $\sigma$ induces an isometric suffix-encoding operator
\begin{equation}
    {S_j}
    := \sum_{\omega \in \mathcal{B}_{V_j}} \ket{\sigma_j(\omega)}\bra{\omega}, \label{Suffix_Encoder_Operator}
\end{equation}
such that we can write
\begin{equation}
    \ket{\sigma_j(v)}={S_j}\ket{v}
\end{equation}
for all $\ket{v} \in \mathcal{V}_j$, mapping each node to its intra-cluster suffix address.
By hierarchically composing the prefix and suffix encoders, we obtain the complete hierarchical encoder, formulated as follows.

%==================================================
% Although the source message is fixed and represents the intrinsic identity of node, the assigned codeword may be changed dynamically by redefining the encoding isometry. This enables dynamic addressing while preserving the fixed source message space.
%=========================================
\begin{proposition}[Hierarchical Encoder]\label{proposition encoder}
Let ${K}$ and ${S}_j$ are defined as in (\ref{Prefix_Encoder_Operator}) and (\ref{Suffix_Encoder_Operator}) respectively.
Then, we have a unique isometry map
\[ \hat{c} : \mathcal{V} \to \mathcal{A}_{\mathrm{e}} \otimes\left( \bigoplus_{j=1}^{n_{\mathrm{e}}} \mathcal{A}_j\right)\] 
given by the hierarchical encoder
\[
    \widehat{C} = \bigoplus_{j=1}^{n_\mathrm{e}} \left(\ket{\kappa_j}\bra{\kappa_j}   \,\otimes\, {S}_j\right),
\]
that satisfies
$\widehat{C}\ket{v} = \ket{\hat{c}(v)}=\ket{\kappa_j \sigma_{ij}}$ for all $v \in V$.
\end{proposition}

\begin{proof}
Let fix an arbitrary cluster $j$. By construction, for any $j$ and $v \in V_j$ we want to have the codeword 
\[
    \ket{\hat{c}(v)}  = \ket{\kappa_j \sigma_j(v)}
                =\ket{\kappa_j} \otimes \ket{\sigma_j(v)}.
\]
On the subspace $\mathcal{V}_j$, the suffix is implemented by any isometry ${S}_j$, such that
\[
   \ket{\hat{c}(v)} =   \ket{\kappa_j} \otimes {S}_j \ket{v},
    \qquad v \in V_j,
\]
Thus, on $\mathcal{V}_j$, a local encoder that implements this code is given by
\[\widehat{C}_{j}=\ket{\kappa_j}\bra{\kappa_j} \otimes S_j.\]
Since $K$ and $S_j$ are isometries and $\ket{\kappa_j}$ is normalized, it follows that $\widehat C_j$ is an isometry on $\mathcal{V}_j$.

% =======================================
% To verify the isometry property of $\widehat C_j$ on $\mathcal {V}_j$. let for any $\ket{x},\ket{y}\in\mathcal {V}_j$, \[ \langle\widehat C_j x,\widehat C_j y\rangle =\langle\kappa_j|\kappa_j\rangle;\langle S_j x,S_j y\rangle =1\cdot\langle x,y\rangle, \] since $\langle\kappa_j|\kappa_j\rangle=1$ and $S_j^\dagger S_j=I_{\mathcal V_j}$. Equivalently, \[ \widehat C_j^\dagger\widehat C_j =\bigl(\ket{\kappa_j}\bra{\kappa_j}\bigr)\otimes (S_j^\dagger S_j) =\ket{\kappa_j}\bra{\kappa_j}\otimes I_{\mathcal V_j}, \] which acts as the identity on the encoded copy $\ket{\kappa_j}\otimes\mathcal {V}_j$. Therefore $\widehat C_j$ is an isometry on $\mathcal{V}_j$.
% =======================================
Next, the source space decomposes as a direct sum of orthogonal subspaces as in (\ref{SourceSpace_decomposition}),
%\[\mathcal{V} = \bigoplus_{j=1}^{n_\mathrm{e}} \mathcal{V}_j,\]
which in turn induces the following decomposition of the address space:
\[
\widehat{\mathcal{A}} 
= \bigoplus_{j=1}^{n_\mathrm{e}} \widehat{\mathcal{A}}_j, 
\quad \text{with} \quad
\widehat{\mathcal{A}}_j := \ket{\kappa_j} \otimes \mathcal{A}_j.
\]
Consequently, the global hierarchical encoder must be block-diagonal with respect to this decomposition, and hence it can be written as
\[
    \widehat{C}
    =
    \bigoplus_{j=1}^{n_\mathrm{e}}
    \left(
        \ket{\kappa_j}\bra{\kappa_j}
        \otimes
        S_j
    \right).
\]
Further, since each ${S}_j$ is an isometry and $\ket{\kappa_j}\bra{\kappa_j}$ is an orthogonal projector, their direct sum defines an isometry. By construction this operator satisfies $\widehat{C}\ket{v} = \ket{\hat{c}(v)}$ for all $v \in V$.
\end{proof}

%==========================================
\subsection{Decoder and Unique Decodability}
%======================================================================
Since the hierarchical encoder $\widehat{C}$ is an isometry from 
$\mathcal{V}$ into $\widehat{\mathcal{A}} \subset \mathcal{H}^{\oplus \ell_{\max}}$, its inverse is well defined on its range. Within the framework of quantum information theory, the inverse of an isometric encoding is given by the adjoint operator restricted to the corresponding code subspace \cite{Wilde_2017}. 
Thus, the hierarchical decoder of $\widehat{C}$ can be formulated as follows.
%======================================================================
\begin{proposition}[Hierarchical Decoder]\label{proposition decoder}
Let $\widehat{C} : \mathcal{V} \to \widehat{\mathcal{A}}$ be the hierarchical encoder of Proposition~\ref{proposition encoder}, with
\[
  \mathcal{V} = \bigoplus_{j=1}^{n_{\mathrm{e}}} \mathcal{V}_j,
  \qquad
  \widehat{\mathcal{A}}
    = \bigoplus_{j=1}^{n_{\mathrm{e}}} \widehat{\mathcal{A}}_j,
  \qquad
  \widehat{\mathcal{A}}_j := \ket{\kappa_j} \otimes \mathcal{A}_j.
\]
For each cluster $j$, define the local decoder
\begin{align*}
\widehat{D}_j: \widehat{\mathcal{A}}_j \to \mathcal{V}_j, \qquad 
  \widehat{D}_j
    := \ket{\kappa_j}\bra{\kappa_j} \otimes {S}_j^\dagger    
\end{align*}
Then the global hierarchical decoder
\[
  \widehat{D} : \widehat{\mathcal{A}} \to \mathcal{V}, \qquad
    \widehat{D} := \bigoplus_{j=1}^{n_{\mathrm{e}}} \widehat{D}_j
    = \widehat{C}^\dagger    
\]
satisfies $\widehat{D}\,\widehat{C} = I_{\mathcal{V}},$
and the hierarchical code is uniquely decodable both locally and globally.
\end{proposition}
%======================================================================
\begin{proof}
Fix a cluster $j$. On $\mathcal{V}_j$, the encoder of Proposition~\ref{proposition encoder} is given by
\[
  \widehat{C}_j: \mathcal{V}_j \to \widehat{\mathcal{A}}_j \qquad
    \widehat{C}_j := \ket{\kappa_j}\bra{\kappa_j} \otimes S_j,
\]
where $S_j$ is an isometry. Consequently,
\begin{align*}
    \widehat{D}_j \widehat{C}_j
    &= (\ket{\kappa_j}\bra{\kappa_j} \otimes S_j^\dagger)
      (\ket{\kappa_j}\bra{\kappa_j} \otimes S_j)\\
    &= (\ket{\kappa_j}\bra{\kappa_j})^2 \otimes S_j^\dagger S_j\\
    &= \ket{\kappa_j}\bra{\kappa_j} \otimes I_{\mathcal{V}_j}.
\end{align*}
Therefore, upon restriction to the subspace $\mathcal{V}_j$ we obtain
\[
  \widehat{D}_j \widehat{C}_j = I_{\mathcal{V}_j}.
\]
In particular, for every $\ket{v} \in \mathcal{V}_j$ it follows that
\[
  \widehat{D}_j \widehat{C}_j \ket{v} = \ket{v},
\]
which demonstrates the decodablity on each cluster subspace $\mathcal{V}_j$.
Using the direct-sum decompositions
\[
  \mathcal{V} = \bigoplus_j \mathcal{V}_j,
  \qquad
  \widehat{\mathcal{A}} = \bigoplus_j \widehat{\mathcal{A}}_j,
\]
the global encoder and decoder decompose as
\[
  \widehat{C} = \bigoplus_j \widehat{C}_j,
  \qquad
  \widehat{D} = \bigoplus_j \widehat{D}_j.
\]
Thus,\[
  \widehat{D}\,\widehat{C}
    = \bigoplus_j \widehat{D}_j \widehat{C}_j
    = \bigoplus_j I_{\mathcal{V}_j}
    = I_{\mathcal{V}},
\]
which proves global decodability and shows that
$\widehat{D} = \widehat{C}^\dagger$ is a left inverse of
$\widehat{C}$ on $\mathcal{V}$. 
In particular, if $\widehat{C}\ket{v} = \widehat{C}\ket{v'}$, then
applying $\widehat{D}$ yields
\[
  \ket{v}
    = \widehat{D}\,\widehat{C}\ket{v}
    = \widehat{D}\,\widehat{C}\ket{v'}
    = \ket{v'},
\]
so the encoder is injective, and hence the hierarchical code is uniquely
decodable. 
\end{proof}
% It is to note that, even if two suffixes coincide in different clusters, the corresponding global codewords lie in orthogonal subspaces
% $\widehat{\mathcal{A}}_j \perp \widehat{\mathcal{A}}_{j'}$ for
% $j \neq j'$, due to the orthogonality of the prefix subspace $\mathcal{A}_{\mathrm{e}}$, and therefore cannot collide.

\begin{figure}
    \centering
   \includegraphics[width=\columnwidth, height=0.3\textheight, keepaspectratio]{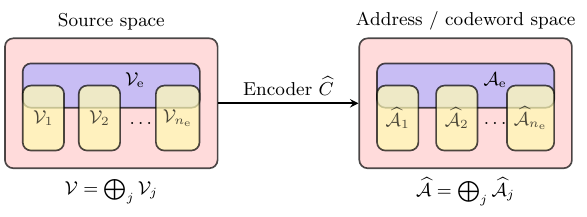}
    \vspace{-0.5cm}
    \caption{Hierarchical decomposition of the source space and its image under isometric encoder $\widehat{C}$, which maps each subspace $\mathcal{V}_j$ into mutually orthogonal address subspace $\widehat{\mathcal{A}}_j$}.
    \label{fig:encoder_mapping}
\end{figure}

These encoder and decoder appear to resemble the Ahlswede and Cai scheme \cite{ahlswede2003lossless} in that the source space is partitioned into subspaces, and each equipped with its own local isometry. In that scheme, the information regarding the subspace to which the source message belongs must be transmitted through a classical side channel. Our construction differs by embedding this information directly into the prefix of codeword address, as illustrated in Fig~\ref{fig:encoder_mapping}. 
In addition, the hierarchical encoder and decoder described above naturally extends to multiple layers. The suffix code of one layer may be treated as the prefix code of the next, inducing a further direct‑sum decomposition of the corresponding code subspaces. This recursive construction leads to a multi‑layer hierarchical code, but a full treatment lies beyond the scope of the present paper.

% %======================================================================
% \subsection{Address Average Length Bound}
% It is reasonable to assume that each cluster and each node has distinct traffic statistics, from which we can estimate the probability that each sends a packet. We can then assign short suffix address for node inside cluster that sends packet with high probability and assign longer address for node with low probability. The same we can also do with the prefix address, where probability now corresponds to the size of the cluster.
% -probability distribution $p: V \rightarrow [0, 1]$.
% -Altogether, we have the probability $p(i,j)=p(i|j) p(j)$.
% -Is the lower bound < $\text{log}_2|V|$?.
% -Compare \cite{ahlswede2003lossless}

%======================================================================
 \section{Encoding Procedure} \label{Section encoding procedure}
 %======================================================================
Having presented the construction of the hierarchical encoder and decoder, we now proceed to describe a procedure that we propose for their implementation. The procedure is enabled by considering the following key observations.
\begin{observation}[Visible coding]\label{observation visible coding}
Since each node is a distinct physical entity, we know which node $v\in V$ corresponds to $\ket{v}\in\mathcal{V}$ when preparing the source message. Therefore, the encoder operates in the visible coding regime, where the output state of the source message to be encoded is known a priori \cite{HowardBarnum_2001, PhysRevA.65.032313}.
\end{observation}
In our setting, visible coding allows the prefix and suffix code, $\kappa$ and $\sigma_j$, to perform one-to-one mapping between the source message and the corresponding prefix and suffix addresses. For such a mapping, we use a prefix-free code in our encoding procedure, motivated by its desirable properties given in the following observation.
\begin{observation}[Prefix-free code]\label{observation prefix free}
A prefix-free code assigns codewords in such a way that no codeword is a prefix of another. When embedded as basis in the quantum bitstring space $\mathcal{H}^{\oplus \ell_{\mathrm{max}}}$, these codewords form a prefix-free set in the sense of Definition~3.1 of
Müller--Rogers~\cite{Mller2008QuantumBS}. In particular, they occupy mutually orthogonal subspaces and therefore span a prefix-free Hilbert subspace.
\end{observation}
Moreover, prefix-free codewords have an additional structural benefit for our hierarchical
construction, as given in the follwoing observation.
\begin{observation}[Concatenation prefix-free codewords]\label{observation concatenation}
Prefix-free codewords have definite length and are therefore length eigenstates
in $\mathcal{H}^{\oplus \ell_{\mathrm{max}}}$. By Lemma~2.5 of
Müller--Rogers~\cite{Mller2008QuantumBS}, the concatenation of two
variable-length quantum codewords is an isometry whenever the first component
is a length eigenstate.
\end{observation} 
This observation is consistent with Proposition~\ref{proposition encoder}, where we proved that the structure of the hierarchical encoder $\widehat{C}$ preserves the isometry. Hence, by implementing prefix-free codewords, Observation~\ref{observation concatenation} confirms that for each cluster $j$, the concatenated map $\ket{\kappa_j}\bra{\kappa_j}\otimes \widehat{S}_j$ is an isometry.

With those observations in place, we propose the following encoding procedure. It may be viewed as an extension of the Bostrom \cite{PhysRevA.65.032313} and Ahlswede \cite{ahlswede2003lossless} scheme, in which we incorporate a hierarchical structure and, in addition, modify the underlying mapping by embedding a prefix-free code, such as Huffman code, into the basis codewords. In particular, under the Huffman scheme, the source probability naturally corresponds to the eigenvalues of the source message matrix in the quantum setting. Once the eigenvalues are treated as classical probabilities, a standard Huffman algorithm can be used to generate classical codewords. This allows for a one-to-one mapping where each Huffman codeword is assigned to the associated eigenstate \cite{850709}.
\begin{procedure}[Huffman-based Hierarchical Encoding]\label{procedure_hirearchical-coding}
Suppose the ensemble $\{p(v),\ket{v}\}_{v\in V}$ represents a discrete quantum source.

\begin{enumerate}[label=\textbf{\arabic*)}]
% ---------------------------------------------------------
\item \textbf{Prefix stage.}

\begin{enumerate}[label=(\alph*)]
\item Diagonalize the local ESP source matrix
\[
\rho_\mathrm{e}=\sum_{v\in\mathcal{V}_\mathrm{e}} p(v)\ket{v}\bra{v}
\]
and use its eigenvectors $\{\ket{\omega_k}\}_{k=1}^{d_{\mathrm{e}}}$ as the orthonormal basis of $\mathcal{V}_{\mathrm{e}}\subset \mathcal{V}$.

\item Let $\{\lambda_k\}_{k=1}^{d_{\mathrm{e}}}$ denote the corresponding eigenvalues of $\rho_\mathrm{e}$. Use these eigenvalues as the canonical probabilities for coding:
\[
p(\omega_k):=\lambda_k=\bra{\omega_k}\rho_\mathrm{e}\ket{\omega_k},
\]
where $k=1,\dots,d_{\mathrm{e}}$.
\item Construct a binary Huffman code
\[
h:\{1,\dots,d_{\mathrm{e}}\}\rightarrow\{0,1\}^{\ell_k}
\]
of variable length $\ell_k$ for the probability distribution \(\{\lambda_k\}\), and represent each outcome by the quantum state \(\ket{h(\omega_k)}\).
\item Define the prefix isometry
\[
K=\sum_{k=1}^{d_{\mathrm{e}}}\ket{h(\omega_k)}\bra{\omega_k}.
\]
For each prefix node, express
$\ket{v_{\mathrm{e}j}}=\sum_k \alpha_{jk}\ket{\omega_k}$ so that 
$\ket{\kappa_j}=K\ket{v_{\mathrm{e}j}}$.
\end{enumerate}

% ---------------------------------------------------------
\item \textbf{Suffix stage.}
\begin{enumerate}[label=(\alph*)]
\item Choose $\ket{v_{e_j}}$ whose corresponding $\ket{c(v_{\mathrm{e}j})}$ is a length eigenstate. Then, set
$\mathcal{V}_j=\mathrm{span}\bigl(\{\ket{v_{\mathrm{e}j}}\}\cup\mathcal{V}_{\mathrm{u}j}\bigr)$ for a prescribed partition $\mathcal{V}_{\mathrm{u}j}\subseteq\mathcal{V}\setminus\mathcal{V}_{\mathrm{e}}$, where
$j=1,\dots,n_{\mathrm{e}}$.
\item For each cluster $j$, form the local source matrix
\(\rho_j=\sum_{v\in\mathcal{V}_j} p(v)\ket{v}\bra{v}\) and diagonalize it. Use the eigenvectors
\(\{\ket{\varphi_{j,k}}\}_{k=1}^{d_j}\) as the orthonormal basis of \(\mathcal{V}_j\).
\item Let \(\{\lambda_{j,k}\}_{k=1}^{d_j}\) be the eigenvalues of \(\rho_j\). Use these eigenvalues as the canonical probabilities:
\[
p(\varphi_{j,k}):=\lambda_{j,k}=\bra{\varphi_{j,k}}\rho_j\ket{\varphi_{j,k}},
\]
where $k=1,\dots,d_j$.
\item Construct a local binary Huffman code \(h_j\) for the distribution \(\{\lambda_{j,k}\}\)
and represent each outcome by the quantum length eigenstate \(\ket{h_j(\varphi_{j,k})}\).
\item Define the suffix isometry
\[
S_j=\sum_{k=1}^{d_j}\ket{h_j(\varphi_{j,k})}\bra{\varphi_{j,k}}.
\]
\end{enumerate}
% ---------------------------------------------------------
\item \textbf{Hierarchical encoding.}
For any $\ket{v}\in\mathcal{V}_j$,
\[
\ket{\widehat{c}(v)}
=\widehat{C}_j\ket{v}
=\ket{\kappa_j}\otimes S_j\ket{v}.
\]
% ---------------------------------------------------------
\item \textbf{Hierarchical decoding.}
Decoding is performed coherently by the adjoint isometry
\[
\widehat{C}^\dagger
=\bigoplus_{j=1}^{n_{\mathrm{e}}}\widehat{C}_j^\dagger,
\]
which reconstructs the state of original source message via 
$\ket{v}=\widehat{C}^\dagger\ket{\widehat{c}(v)}$.
\end{enumerate}
\end{procedure}
%===============================================

%===============================================
The procedure described above is only one of several possible approaches. In our case, when a local density operator \(\rho\) possesses degenerate eigenvalues, we may arbitrarily select any orthonormal basis within each degenerate eigenspace. Alternatively, for example, a Gram–Schmidt procedure can be employed when starting from a set of non-orthogonal source messages, as suggested in \cite{PhysRevA.65.032313}. 

Furthermore, any other prefix-free code can also be employed effectively. Our choice of the Huffman code is motivated by its optimality in attaining the Shannon entropy in the classical case \cite{gersho2012vector}. Depending on the statistical distribution of the source messages (traffic), alternative prefix codes may be more advantageous. For instance, the Golomb code is particularly well suited when the source messages follow a geometric distribution. Several implementation strategies are available for realizing the encoders $K$ and $S_j$, particularly to accommodate the space $\mathcal{H}^{\oplus\ell_{\mathrm{max}}}$, whose elements are vectors of variable length. One approach is the zero‑extended form (ZEF) of Schumacher and Westmoreland \cite{SchumacherWestmoreland2001}, which embeds indeterminate‑length strings into fixed‑length registers by padding with zeros. An alternative is to embed states into the tape Hilbert space of a quantum Turing machine (QTM) and implement coherent write/read operations (prefix parsing or delimiters) on the tape \cite{Mller2008QuantumBS}.

One of useful features in our procedure, the source message can be coherently reconstructed without the need for any projective measurements due to the isometry. Moreover, in contrast to the schemes proposed by Boström \cite{PhysRevA.65.032313} and Ahlswede \cite{ahlswede2003lossless}, our construction does not require any classical communication of codeword lengths or subspace indices for decoding the address codewords. This follows from the fact that the prefix, which is appended to each suffix, already encodes the subspace indices and is chosen to be an orthogonal length eigenstate. As investigated by Müller and Rogers \cite{Mller2008QuantumBS}, the concatenation for arbitrary non‑length‑eigenstate codewords remains an open problem. Thus, for the time being, we leave the cluster empty of user nodes if the prefix is not a length eigenstate to avoid operational complexity. The expense of this construction is a likely increase in the total codeword length due to the appending of the prefix.
%The upper bound of the average length is still under investigation.
 %==============================================================

 %===================================================================================================================
\section{Numerical Example} \label{numerical_example}
%===================================================================================================================
To demonstrate the effectiveness of hierarchical coding, we apply Procedure~\ref{procedure_hirearchical-coding} to a small network as a toy example. We consider a network comprising 13 nodes. Among these, 3 nodes function as ESP nodes within the tier‑2 layer. All nodes are partitioned into 3 clusters, each managed by a single ESP. In the tier‑1 layer, we specify that each cluster includes 6, 4, and 0 user nodes, respectively. We model the network nodes as a discrete information source specified by the ensemble \(V=\{p,\mathcal V\}\). Each cluster is modeled as an independent source, so the joint distribution factorizes as
\[
P(V)=\prod_{j=1}^{3}P(V_j),
\]
which is valid because the cluster preparations are independent and no cross‑cluster quantum correlations are present. Consequently, the ensemble decomposes into independent cluster components and separate prefix and local suffix codes are justified.

\paragraph{Prefix stage}
Suppose the ESP node probabilities are
\[
p(v_{\mathrm{e}1})=0.5,\quad p(v_{\mathrm{e}2})=0.3,\quad p(v_{\mathrm{e}3})=0.2.
\]
We work in a two‑dimensional Hilbert space \(\mathcal V_{\mathrm e}=\mathrm{span}\{\ket{0},\ket{1}\}\) and have prepared the ESP states
\[
\ket{v_{\mathrm{e}1}}=\ket{0},\quad \ket{v_{\mathrm{e}2}}=\ket{1},\quad
\ket{v_{\mathrm{e}3}}=\tfrac{1}{\sqrt{2}}(\ket{0}+\ket{1}).
\]
The local ESP source matrix is
\[
\rho_{\mathrm e}=\sum_{j=1}^3 p(v_{\mathrm{e}j})\ket{v_{\mathrm{e}j}}\bra{v_{\mathrm{e}j}}
=\begin{bmatrix}0.6 & 0.1\

\\0.1 & 0.4\end{bmatrix}.
\]
%with von Neumann entropy \(S(\rho_{\mathrm e})\approx 0.9403\).
Following Procedure~\ref{procedure_hirearchical-coding}, we diagonalize \(\rho_{\mathrm e}\), and obtain the eigenpairs
\[
\lambda_1\approx 0.641,\quad \ket{\omega_1}\approx\begin{bmatrix}0.924\
\\0.383\end{bmatrix},
\]
\[
\lambda_2\approx 0.358,\quad \ket{\omega_2}\approx\begin{bmatrix}-0.383\
\\0.924\end{bmatrix},
\]
with \(\langle\omega_1|\omega_2\rangle=\delta_{12}\). We then use the eigenvalue distribution \(\{\lambda_1,\lambda_2\}\) as the Huffman probabilities. For two symbols the natural prefix assignment is
\[
\ket{h(\omega_1)}=\ket{0},\qquad \ket{h(\omega_2)}=\ket{1},
\]
and the prefix isometry is
\[
K=\ket{0}\bra{\omega_1}+\ket{1}\bra{\omega_2}.
\]
% Note that \(K\) is not the identity in the computational basis because \(\{\ket{\omega_k}\}\) are not the computational basis vectors.
The address of \(v_{\mathrm{e}3}\) is obtained by $\ket{k(v_{\mathrm{e}3})}=K\ket{v_{\mathrm{e}3}}$.

%=======================================================================
\begin{table*}
    \centering
    \caption{The nodes under consideration grouped by cluster, with their address assignment.}
    \label{source_messages table}
    \begin{tabular}{c p{0.025\linewidth} p{0.025\linewidth} p{0.025\linewidth} p{0.15\linewidth} p{0.375\linewidth} c}
    \toprule
       Cluster  &  Node &  $p(v)$ & $p(v|j)$ & $\ket{v}$& $\ket{\hat{c}(v)}$& Fidelity \\
    \midrule
    1&$v_{11}$ & 0.15 &0.30& $\ket{00}$& $0.56\ket{0010}+0.37\ket{0011}-0.64\ket{0100}+0.06\ket{0110}+0.23\ket{1010}+0.15\ket{1011}-0.26\ket{1100}+0.02\ket{1110}$& 1\\
    1&$v_{21}$ & 0.07 & 0.14& $\ket{01}$& $-0.71\ket{0010} + 0.08\ket{0011} - 0.58\ket{0100} - 0.13\ket{0110} - 0.29\ket{1010} + 0.03\ket{1011} - 0.24\ket{1100} - 0.05\ket{1110}$& 1\\
    1&$v_{31}$ & 0.06& 0.12& $\ket{10}$& $0.14\ket{0010} - 0.82\ket{0011} - 0.33\ket{0100} + 0.22\ket{0110} + 0.06\ket{1010} - 0.34\ket{1011} - 0.14\ket{1100} + 0.09\ket{1110}$& 1\\
    1&$v_{41}$ & 0.06& 0.12& $\ket{11}$& $-0.17\ket{0010} + 0.19\ket{0011} + 0.04\ket{0100} + 0.89\ket{0110} - 0.07\ket{1010} + 0.08\ket{1011} + 0.02\ket{1100} + 0.37\ket{1110}$& 1\\
    1&$v_{51}$ & 0.05& 0.10& $0.71\ket{01}+0.71\ket{11}$& $-0.62\ket{0010} + 0.19\ket{0011} - 0.38\ket{0100} + 0.54\ket{0110} - 0.26\ket{1010} + 0.08\ket{1011} - 0.16\ket{1100} + 0.22\ket{1110}$& 1\\
    1&$v_{61}$ & 0.05& 0.10& $0.71\ket{10}+0.71\ket{11}$& $-0.03\ket{0010} - 0.45\ket{0011} - 0.21\ket{0100} + 0.78\ket{0110} - 0.01\ket{1010} - 0.18\ket{1011} - 0.09\ket{1100} + 0.32\ket{1110}$& 1\\
    1&$v_{71}$ & 0.06& 0.12& $0.50\ket{00}+0.50\ket{01}+0.50\ket{10}+0.50\ket{11}$& $-0.09\ket{0010} - 0.09\ket{0011} - 0.75\ket{0100} + 0.52\ket{0110} - 0.04\ket{1010} - 0.04\ket{1011} - 0.31\ket{1100} + 0.21\ket{1110}$& 1\\
    \midrule
    2& $v_{12}$& 0.08 & 0.267&$\ket{00}$& $-0.27\ket{0010} - 0.27\ket{0110} + 0.65\ket{1010} + 0.65\ket{1110}$& 1\\
    2& $v_{22}$ & 0.06& 0.2& $\ket{01}$& $0.27\ket{0011} - 0.27\ket{0100} - 0.65\ket{1011} + 0.65\ket{1100}$& 1\\
    2& $v_{32}$ & 0.05& 0.167& $\ket{10}$& $-0.27\ket{0010} + 0.27\ket{0110} + 0.65\ket{1010} - 0.65\ket{1110}$& 1\\
    2& $v_{42}$ & 0.06& 0.2& $\ket{11}$& $-0.27\ket{0011} - 0.27\ket{0100} + 0.65\ket{1011} + 0.65\ket{1100}$& 1\\
    2& $v_{52}$ & 0.05& 0.167& $0.71\ket{01}+0.71\ket{11}$& $-0.38\ket{0100} + 0.92\ket{1100}$& 1\\
    \midrule
    3& $v_{13}$& 0.2 & 1& $0.71\ket{0}+0.71\ket{1}$&$0.71\ket{0}+0.71\ket{1}$& 1 \\
    \bottomrule
    \end{tabular}    
\end{table*}
%=======================================================================

\paragraph{Suffix stage (cluster 1)}
Assume cluster 1 has conditional probabilities as specified in Table \ref{source_messages table}. We use a four‑dimensional Hilbert space
\[\mathcal V_1=\mathrm{span}\{\ket{00},\ket{01},\ket{10},\ket{11}\}\]
and prepare the node states (including the ESP) as source messages $\ket{v}$ in Table \ref{source_messages table}. The local source matrix is computed as
\begin{align*}   
\rho_1&=\sum_{i=1}^7 p(v_{i1}\!\mid\!v_{\mathrm e1})\ket{v_{i1}}\bra{v_{i1}} \\
&\approx
\begin{bmatrix}    
0.33 & 0.03 &  0.03 &  0.03 \\
0.03 & 0.22 &  0.03 &  0.08 \\
0.03 & 0.03 &  0.20 &  0.08 \\
0.03 & 0.08 &  0.08 &  0.25 
\end{bmatrix}.
\end{align*}
We diagonalize \(\rho_1\) and let \(\{\lambda_{k,1},\ket{\varphi_{k,1}}\}_{k=1}^{d_1}\) denote its eigenpairs. Using the spectral rule we take the eigenvalues as the Huffman probabilities. For this instance the diagonalization yields (rounded)
\[
\lambda_{1,1}\approx0.40,\, \lambda_{2,1}\approx0.29,\, \lambda_{3,1}\approx0.18,\, \lambda_{4,1}\approx0.13,
\]
which sum to unity up to rounding, and the orthonormal basis vectors $\ket{\varphi_{1k}}$ for $\mathcal{V}_1$ stacked as the columns of matrix 
\begin{align*}
    \Phi_1=
    \begin{bmatrix}
    0.606 &  -0.794 & -0.004 &  0.053 \\
    0.419 &   0.296 & -0.754 & -0.410 \\
    0.382 &   0.247 &  0.643 & -0.616 \\
    0.558 &   0.470 &  0.130 &  0.670
\end{bmatrix}.
\end{align*}
Next, we construct a local Huffman code \(h_1\) that assigns a shorter codeword to the larger eigenvalue. One valid assignment is
\begin{align*}
    &\ket{h_1(\varphi_{1,1})}=\ket{1},\quad &\ket{h_1(\varphi_{1,2})}=\ket{01},\quad \\
    &\ket{h_1(\varphi_{1,3})}=\ket{000},\quad &\ket{h_1(\varphi_{1,4})}=\ket{001},
\end{align*}
and the suffix isometry is
\begin{align*}    S_1&=\ket{0}\bra{\varphi_{1,1}}+\ket{01}\bra{\varphi_{1,2}} \\
&+\ket{000}\bra{\varphi_{1,3}}+\ket{001}\bra{\varphi_{1,4}}.
\end{align*}
For the implementation of the computation of the encoder $S_1$, we employ the ZEF scheme \cite{PhysRevA.65.032313, SchumacherWestmoreland2001}, which is selected due to its relatively straightforward and practically tractable realization. It provides a finite-dimensional Hilbert space and simple matrix representations of the encoders. In addition, we may introduce a row-selection (permutation) matrix \(P\) to aggregate the utilized padded indices while preserving the isometric property. In this setting, we can express
\begin{align*}
    S_1 = P\,\Gamma_1 \Phi_1^{\dagger},
\end{align*}
where \(\Gamma_1\) is constructed such that its columns are given by the ZEF of the basis codewords \(\ket{h(\varphi_{1,k})}\).

\paragraph{Suffix stage (cluster 2)}
In an analogous manner, we can compute the hierarchical encoder and decoder for cluster 2. Here, we choose the same basis as in cluster 1 to describe the source messages in cluster 2. With the probabilities listed in Table~\ref{source_messages table}, we obtain the corresponding source matrix
\begin{align*}   
\rho_2&=\sum_{i=1}^5 p(v_{i2}\!\mid\!v_{\mathrm e2})\ket{v_{i2}}\bra{v_{i2}} \\
&\approx
\begin{bmatrix}    
0.27 & 0 &  0 &  0 \\
0 & 0.28&  0&  0.08 \\
0 & 0&  0.17 &  0 \\
0 & 0.08 &  0&  0.28 
\end{bmatrix}.
\end{align*}
Accordingly, the spectral decomposition of the source matrix $\rho_2$ yields its associated eigenvalues, given by
\[
\lambda_{1,2}\approx0.37,\, \lambda_{2,2}\approx0.27,\, \lambda_{3,2}\approx0.20,\, \lambda_{4,2}\approx0.17,
\]
and its eigenvectors, stacked as columns in the matrix $\Phi_2$, given by
\begin{align*}
    \Phi_2=
    \begin{bmatrix}
    0 &  1 & 0 &  0 \\
    0.70 &   0 & 0.70 & 0 \\
    0 &   0 &  0 & -1 \\
    0.70 &   0 &  -0.70 &  0
\end{bmatrix}.
\end{align*}
As in cluster 1, we assign a Huffman codeword to each basis vector based on its eigenvalue distribution to obtain the suffix encoder
\begin{align*}    S_2&=\ket{0}\bra{\varphi_{2,1}}+\ket{01}\bra{\varphi_{2,2}} \\
&+\ket{000}\bra{\varphi_{2,3}}+\ket{001}\bra{\varphi_{2,4}}.
\end{align*}

\paragraph{Hierarchical encoding}
For any \(\ket{v}\in\mathcal V_j, \text{ with } j=1,2\) , we  obtain the address codeword
\[
\ket{\widehat c(v)}=\ket{\kappa_j}\otimes S_j\ket{v},
\quad 
\]
as provided in Table~\ref{source_messages table}. %Explain the resulted address, no need of orthogonality, dynamic assignment.
As shown in the table, the encoded addresses are represented as a superposition of 4-qubit strings that are constructed from the tensor product of the basis codewords assigned to the prefix and suffix. The encoded prefix ensures that codewords originating from different clusters lie in orthogonal subspaces, thereby preventing overlaps. This is shown by the evidence that even identical source messages appearing in both clusters are mapped to distinct address codewords and can be uniquely decoded. Meanwhile, within each cluster, the source messages and their corresponding codewords do not need to be orthogonal. The only requirement is an isometric encoder on the entire cluster subspace.

\paragraph{Decoding and Evaluation}
The source message $\ket{v}$ can be recovered from the address codeword by applying the decoder $D=\widehat C^\dagger$, computed as
\[\ket{v}=\widehat C^\dagger\ket{\widehat c(v)}.\] For each source message $\ket{v}$, we evaluate its fidelity  as
\[
F(v) =  \vert\bra{v}\widehat{C}^\dagger \ket{\widehat{c}(v)}\vert^2,
\]
which serves as a quantitative measure of the accuracy of the proposed address coding scheme. As shown in Table~\ref{source_messages table}, it is possible to attain fidelity 1, which confirms the degree to which the representation is lossless. %Discuss if there exist noise and QEC?
However, it is important to note that we examined this address coding in an ideal, noise-free environment. An extension to a noisy setting would be an interesting direction for future work, for example, by equipping the address with an appropriate quantum error-correcting scheme. 

In addition, we also asses the isometry error of any encoder $X$ by computing
\[\
\epsilon_{\mathrm{iso}}(X):=\Vert X^\dagger X - I\Vert_2.
\]
As shown in Table \ref{numerical_validation table}, each encoder preserves its isometry with only a negligible error due to finite numerical precision. Altogether, the numerical evaluation confirms the isometry property stated in Proposition~\ref{proposition encoder}, as well as the decodability guarantee established in Proposition~\ref{proposition decoder}.

\begin{table}[t]
\centering
\caption{Numerical Evaluation of the hierarchical encoders.}
\label{numerical_validation table}
\begin{tabular}{c c c}
\hline
Encoder & Size & Isometry error \\
\hline
$K$& $2\times 2$& $8.48 \times 10^{-8}$\\
$S_1$ & $8\times 4$ & $5.87\times 10^{-16}$  \\
$S_2$& $8\times 4$ & $6.14\times 10^{-16}$  \\
$\widehat{C}_1$& $16\times 4$&$2.54 \times10^{-7}$\\
$\widehat{C}_2$& $16\times 4$& $2.54 \times10^{-7}$\\
$\widehat{C}$& $16\times 8$& $1.70\times10^{-7}$\\
\hline
\end{tabular}
\end{table}

\section{Conclusion}\label{conclusion}
%===================================================================================================================

We have shown that addressing in quantum networks is naturally framed as a lossless coding problem. By introducing a hierarchical prefix–suffix address space and implementing isometric, prefix‑free encoders and their adjoints, we ensure that addresses are uniquely decodable and compatible with coherent processing. Our Huffman‑based procedure demonstrates how the address structure can be exploited to produce length‑eigenstate codewords, thereby avoiding classical side channels for subspace indices. The design explicitly accommodates heterogeneous cluster sizes and dynamic reconfiguration, allowing address assignment to be adjusted as the entanglement topology evolves. Future work will analyze performance under realistic noise models, explore multi‑layer recursive hierarchies, and investigate their use for different network functionalities. Overall, lossless address coding provides a principled, scalable foundation for addressing in next‑generation quantum networks.

\IEEEpeerreviewmaketitle

\section*{Use of AI Disclosures}
%%===================================================================================================================
% This work has been funded by the European Union under Horizon Europe
% ERC-CoG grant QNattyNet, n.101169850. Views and opinions expressed are
% however those of the author(s) only and do not necessarily reflect those of the European Union or the European Research Council Executive Agency. Neither
% the European Union nor the granting authority can be held responsible for
% them.

Some portions of this manuscript were refined using AI-assisted tools for grammar, clarity, and organization. All AI outputs were reviewed, edited, and verified by the authors. All concepts and technical content are solely those of the authors.



\bibliographystyle{IEEEtran}
\bibliography{References}

\nocite{*} % This includes all entries from the .bib file in the bibliography

%\bibliography{References}

\end{document}